\documentclass[12pt]{amsart}
\usepackage[top=1.2in, left=1.2in, right=1.2in, bottom=1.2in]{geometry}
\usepackage{lmodern}
\usepackage[T1]{fontenc}
\usepackage{microtype}
\usepackage{ucs}
\usepackage[utf8x]{inputenc}
\usepackage{amssymb,amsfonts,amsmath,mathrsfs,color,amsthm,latexsym}
\usepackage[all]{xy}
\usepackage[colorlinks=true]{hyperref}
\usepackage{graphicx}
\usepackage{rotating}
\usepackage{array}

\usepackage{float}
\usepackage{xspace}

\usepackage[english]{babel}

\usepackage{units}

\usepackage{framed}
\usepackage{stmaryrd}

\newcommand{\norm}[1]{\left\Vert#1\right\Vert}

\newtheorem{theorem}{Theorem}[section]

\newtheorem{remark}[theorem]{Remark}
\newtheorem{proposition}[theorem]{Proposition}

\begin{document}
\setcounter{tocdepth}{1}

\title{Second order stochastic differential models for financial markets}

\author{Nguyen Tien Zung}
\address{School of Mathematics, Shanghai Jiao Tong University 
(visiting professor), 
800 Dongchuan Road, Minhang District, Shanghai, 200240 China
and Institut de Math\'ematiques de Toulouse, UMR5219, 
Universit\'e Paul Sabatier, 118 route de Narbonne,
31062 Toulouse, France, tienzung@math.univ-toulouse.fr}

\begin{abstract}
Using agent-based modelling, empirical evidence and physical ideas, such as
the energy function and the fact that the phase space must have twice the dimension 
of the configuration space, we argue that the stochastic differential equations
which describe the motion of financial prices with respect to real world probability 
measures should be of second order (and non-Markovian), instead of first order models
à la Bachelier--Samuelson. Our theoretical result in
stochastic dynamical systems shows that one cannot correctly reduce second
order models to first order models  by simply forgetting about momenta. 
We propose some simple second order models, including a stochastic
constrained n-oscillator, which can explain many market phenomena, 
such as boom-bust cycles, stochastic quasi-periodic behavior, 
and ``hot money'' going from one market sector to another. 
\end{abstract}

\maketitle

\tableofcontents

\section{Introduction}

The main purpose of this paper is to suggest that stochastic differential equations 
and stochastic dynamical systems (SDE's and SDS's)
with respect to ``real world'' probability measure, 
which model financial price movements should be of second order, 
instead of first order Markovian models
à la Bachelier--Samuelson \cite{Bachelier1900,Samuelson_Random1965}
based on Fama's efficient market hypothesis \cite{Fama-EMH1970}.

Of course, the Black--Scholes--Merton option pricing formula \cite{BS1973,Merton_Option1973}, 
which is derived from the first order SDE
\begin{equation} \label{eqn:BS}
\dfrac{dS}{S} 
= \mu dt + \sigma dB_t,
\end{equation}
where $S$ is the price of the underlying asset, $\mu$
is the drift parameter, $\sigma$ is the volatility coefficient, and $B_t$ is a Wiener
process, is a valid formula, modulo some corrections due jumps and fat tails \cite{TankovCont_Jump2003}. However, Equation \eqref{eqn:BS}  itself
is valid only with respect to a risk-neutral measure and not with respect to the
real world probability measure, and it fails to explain many phenomena in the market, 
such as boom-bust cycles. For investing and risk management, one needs more realistic
models than this equation. 

In physics, most important equations of motion are of second oder. 
(The phase space must have twice the dimension of the configuration space, 
e.g., when the configuration space is a manifold $N$ 
then its natural phase space is the tangent or cotangent bundle of $N$).
By analogy, we come to the idea that  stock prices should also be governed by
second order differential equations, with stochastic terms due to noises. In fact, this idea is not new, even though we arrived at it by ourselves: it was already used by 
Cont and Bouchaud \cite{BouchaudCont_LangevinCrash1998}, 
who found it by analogy with 
the classical Langevin's equation, and who used it in their explanation of 
market crashes. Cont and Bouchaud \cite{BouchaudCont_LangevinCrash1998} also said that 
J. Doyne Farmer, who is a pioneer in agent-based modelling of complex systems \cite{Farmer_AgentBased2001,Farmer_Challenge2011-2014}, 
already showed second order differential equations for stock prices in a seminar talk in 1997
in Paris 7, though no written text was publicly available. 

Our paper may be viewed as a further development in the direction
of these ideas of second order SDS and agent-based models. In particular, we contribute the following elements to the theory:

$\bullet$ The notion of \textit{market energy}, in analogy with physical energy, which governs the equations of 
motion of stock prices, and which is responsible for many visible
market phenomena, such as boom-bust cycles and persistence of
volatility. 

$\bullet$ An \textit{agent-based} construction of this market energy, which can be decomposed into many components: kinetic market energy, potential market energy, thermodynamic market energy, etc. 
Each type of market energy corresponds to some kinds of agent behavior and strategy, such as contrarian investing and portfolio rebalancing, momentum players and hedging strategies, etc. 

$\bullet$ A \textit{no-go theorem} in the theory of reduction of stochastic 
dynamical systems, which implies in particular that first order stochastic 
models are \textit{incorrect} reductions of second order models in general.
To get better estimates and predictions in investing and risk managements, 
one needs second  order models. 
(One may guess that many trading shops are actually using them).

$\bullet$ We also discuss some simple second-order stochastic models,  
such as the damped stochastic harmonic oscillator and 
the \textit{constrained $n$-oscillator} model, which are \textit{integrable} 
in the sense of \cite{ZungThien_Stochastic2015}
(read: easy to compute and simulate), and which can 
already capture a lot of features of real-world financial markets, 
e.g., the fact that stocks in the same sector often move together, 
and that ``hot money'' can jump from one stock or sector  to another  
(\textit{market anergy transfer} among the sectors).

Our idea behind the notion of market energy is very simple. In Hamiltonian dynamics, the equation of motion of a  conservative system is written as 
(see, e.g., \cite{Arnold-Mechanics1978})
\begin{equation} \label{eqn:Hamilton}
\dot{x}_i = \dfrac{\partial H}{\partial p_i} \; ; \quad
\dot{p}_i = - \dfrac{\partial H}{\partial x_i},
\end{equation}
where $(x_i)$ are the coordinates on the configuration space (i.e, spatial variables), $(p_i)$ are called the momenta and are the dual coordinates 
of $(x_i)$ on the cotangent bundle to the configuration space, 
and $H = H(p_i,x_i)$ is the total energy  (i.e., the Hamiltonian function). 
The relation between the momenta $(p_i)$ and the velocities $(v_i =  \dot{x}_i)$
is given by the Legendre transformation. For example, in Newtonian mechanics we have
$p_i = m_i v_i = m_i \dot{x_i}$, where $m_i$ is the mass of the particle number $i$.

If one forgets about the momenta $(p_i)$, and consider \eqref{eqn:Hamilton}
as a system of equations of the configuration variables $(x_i)$ only, they
it is a system of second-order differential equations on $(x_i)$. By analogy, in financial market $(x_i)$ can play the role of prices, $H$ is the total market energy to be defined, and Equation \eqref{eqn:Hamilton}, plus some noise terms and external
forces and damping terms, can serve as the second-order differential equation for the movement of the prices. In the simplest case, this equation gives us the damped 
stochastic oscillator (see Section \ref{section:DampedOscillator}), 
which has been extensively studied in mathematics and physics
\cite{Bajaj_Waves2006,Gitterman-NoisyOscillator2005,ZungThien_Stochastic2015}. 
Such a simple oscillator model already fits real-world financial price movements, 
such as USD/EUR exchange rates and inflation-adjusted gold prices over the 
long term much better than first order models.

The rest of this paper is organized as follows:

In Section \ref{section:NoGo} we recall from a recent mathematical paper 
of ours \cite{ZungThien_Stochastic2015} with a former doctoral student 
(who went to work for a bank after his thesis) the no-go theorem for reduction 
of stochastic dynamical system: one cannot reduce the system by 
simply forgetting about some variables. Implication for us: 
the momenta of the stock prices, even though they can not be measured with
confidence, should enter the equations as variables on the same footing as 
stock prices, not just as constants or stochastic parameters. 
In Section \ref{section:MarketEnergy} we discuss the notion of market energy and
how to decompose it into the sum of different components corresponding
to different behaviors/strategies of market players. This market
energy plays the central role in the second order stochastic differential equation for the stock prices. In Section \ref{section:DampedOscillator}
we discuss the damped stochastic oscillator model for
a single stock, and some modifications. 
In Section \ref{section:TA} we discuss some technical analysis patterns which
can be explained by the market energy and other physics-like arguments.
Finally, in Section \ref{section:N-Oscillator} we propose a simple constrained 
$n$-oscillator model for a multi-stock market.

\section{No-go theorem for stochastic reduction}
\label{section:NoGo}

Stochastic differential equations (SDE's) in mathematical finance 
are often written in Ito form, but for global analysis it may be more 
convenient to write them in the following Stratonovich form, which 
can be done in a coordinate-free way and which behaves well under
changes of variables:
\begin{equation}
d x_t = X_0 dt + \sum_{i=1}^k X^i \circ d B^i_t,
\end{equation}
where $x$ denotes a point on a manifold $M$;
$X_0, X_1,\hdots X_k$ are vector fields on $M$, and $B_t^1,\hdots, B_k^t$
are independent Wiener processes, see, e.g., \cite{Bismut-Aleatoire1981,FrWe-Random2012,Hsu_StochasticAnalysis2002, 
Oksendal-SDE2003,ZungThien_Stochastic2015}. This SDE generates a continuous-time stochastic dynamical systems (SDS) whose associated \textit{stochastic vector field} is:
$$\displaystyle X = X_0 + \sum_{i=1}^k X_i \circ \dfrac{dB_{t}^i}{dt},$$
and whose  \textit{diffusion operator} is:
$$A_X = X + \dfrac{1}{2} \sum_1^k X_i^2.$$ 
The meaning of this diffusion operator is as follows: if we fix a point $x \in X$
and denote by $x_t$ the random position of $x$ after time $t$ by the random flow
of the SDS, then for any smooth function $f$ on $M$ we have
$$
(A_Xf)(x) = \lim_{t \to 0} \dfrac{\mathbb{E}^x[f(x_t)] - f(x)}{t}.
$$

Two SDS's on a manifold $M$ are the same if and only if they have the same components for their corresponding stochastic vector fields
up to a permutation. But for probability computations, it is only the
diffusion operator which matters. So we say that two SDS's on $M$
are \textit{diffusion-wise the same} if they have the same diffusion operator.
The proof of the following theorem is straightforward, see \cite{ZungThien_Stochastic2015}:

\begin{theorem}[\cite{ZungThien_Stochastic2015}]
\label{thm:NoGo}
Let $\Phi : M \to N$ be a smooth surjective map from a manifold $M$ to a  manifold $N$, 
and let $\displaystyle X = X_0 + \sum_{i=1}^k X_i \circ \dfrac{dB_{t}^i}{dt}$ 
be an SDS on $M$.
Then the diffusion process of $X$ is projectable (i.e., can be reduced)
to a Markov process on $N$ if and only if for any function $f : N \to \mathbb{R}$ and any two points $x,y \in M$ such that $\Phi(x) = \Phi(y)$ we also have \begin{equation}\label{eq:condition2} A_X(\Phi^*(f)(x)) = A_X(\Phi^*(f)(y))\end{equation} where $A_X$ is the diffusion operator of $X$. If this condition is satisfied and $\Phi$ is a submersion then the projected diffusion process on $N$ is generated by an SDS on $N$. 

In the case when $X= X_0$ is a smooth deterministic system then the deterministic process generated by $X$ on $M$ is projectable to a Markov process on $N$ if and only if for any points $x,y \in M$ such that $\Phi(x) = \Phi(y)$ we also have 
$\Phi_{*}(X(x)) = \Phi_{*}(X(y))$, and if this condition is satisfied then $X$ is projected to a smooth vector field on $N$.
\end{theorem}

In particular, if $M = T^*N$ or $M = TN$ is the cotangent or tangent bundle of $N$, the surjective map $\Phi$ is the projection map, and
our second order model is an SDS on $M$, then in general it is impossible for the conditions of the above theorem to be satisfied,
so we cannot project it to an SDS on $N$ 
(i.e., to a first order model) by simply forgetting about the
momenta. Even if the second order model is deterministic, we still cannot reduce it to a stochastic first order system. So Theorem \ref{thm:NoGo} is a kind of no go theorem for dimensional reduction 
of stochastic models. (Unless when there are some obvious symmetries, in which case the system can be reduced).  What one does in stochastic modelling 
by simply forgetting about some hidden variables is not
really reduction, but rather a kind of rude approximation. 

\section{Market players and market energy}
\label{section:MarketEnergy}

Over the last
three decades, there have been a lot of works on agent-based models of
financial markets (see, e.g., \cite{BouchaudCont_LangevinCrash1998, CMZ_Prototype1997, ContBouchaud_Herd2000, DayHuang_BullBear1990, Farmer_AgentBased2001, Farmer_Challenge2011-2014, FLPPS_AgentStochastic2012, GiaBou_AgentBased203, LeBaron_AgentBased2006, LLS_MicroscopicMarket2000, SZSL_AgentBased2007, TWG_Chartists2015, Voit_StatisticalFinance2010, Westerhoff_Agentbased2009}), and some of them seem to be very successful 
in simulating the real-world markets and for devising market-beating strategies. Nevertheless, we have not
seen the notion of market energy in these works, even though
some market equations introduced there look Hamiltonian-like. 
We believe that this is a very useful notion, not only for our paper,
but also for the other agent-based models as well.

In physics, the energy plays a central role: the equation of motion 
can be derived from the total energy function in Hamiltonian formulation, 
or can be written as a variational (Euler-Lagrange) equation in Lagrangian
formulation, using the action function, which is also a function of energy 
components, and external forces. 

Here we also want to figure out what is the market energy function
which governs the movement of financial prices (together with 
external forces and noises). Like in physics, the market energy
can be decomposed into a sum of several components, such as
follows.

\subsection{Potential market energy}

\textit{Value investors}
tend to buy/sell stocks which they think are undervalued/overvalued. 
Because of that, when an asset price is different from its average perceived fair value level, then the difference between the price and the fair value creates 
a potential net aggregate buying or selling action from the value investors,
i.e. a potential market energy. 

As a first approximation, one can think of this potential energy
as being proportional to the square of the level of mispricing, 
because its derivative (which is the force leading to the change 
in price momentum, because the higher the derivative, 
the more players are going to react) is approximatively
proportional to the level of mispricing. 

It may happen that there are several ``centers of gravity'' for the potential 
energy function, i.e. several different values which can pretend to be ``the
average fair value'', depending on the level of optimism in the market.
Consider, for example, a scenario with two particular groups of value investors:
the optimists and the pessimists. The optimists 
have their perceived fair value o the stock at $V_1$, but they may become
bankrupt or too depressed to buy if the price falls too much below $V_1$.
Conversely, the pessimists  have their perceived 
fair value of the stock at $V_2 << V_1$, but won't have anything to sell if the
price is too much above $V_2$. In this scenario, the potential market energy
function may look like a double well potential, with two bottoms (basins
of attraction)  at $V_1$ and $V_2$.  

One may also think of other scenarios, where the potential energy function is
even more complicated. Nevertheless,
it is safe to assume that this function 
goes to infinity when the price $x$ goes to infinity when the level of 
mispricing goes to infinity.

\subsection{Kinetic market energy} 

\textit{Momentum players}
tend to buy/sell stocks which have shown an increase/decrease in 
prices (i.e. positive/negative momentum, modulo noises). 
This ``buying begets buying'' herd behavior reflects the inertia 
of a stock and can be associated to its kinetic market energy, which is approximatively proportional to the square of the level of net buying
or net selling, and hence is approximatively proportional to the square 
of the momentum. (A transaction is both a buy and a sell, but
will be counted as a buy if done at the ask price and creating an upward pressure 
on the price).  

There are also  investment strategies which are not ``momentum chasing'' 
per se, but which still add to the ``selling begets selling'' kinetic market energy, 
for example the ``performance insurers'' and hedging strategies, which are 
blamed for the 1987 stock market crash (see \cite{KimMarkowitz_InvestmentRules1989, 
LLS_MicroscopicMarket2000}).

\subsection{Thermodynamic market energy} 

It may happen that there is a lot of trading but the price of the stock does not really move
in any direction, except for a random noisy movement like the Brownian motion.
This noise in the market, which consists of a large amount of micro-movements which mostly cancel the direction of each other, may be associated to the market's 
thermodynamic energy (heat) and is responsible for the stochastic term in the 
price motion equation. 

Those market makers who provide liquidity for the market without betting on its direction one way or another may be considered as contributing the heat 
to the market as well. There are always 
energy-losing damping forces on the market (such as transaction costs, bid/ask spreads), so the mechanical energy level of the market may slowly go to zero (i.e., the market dies out) if there is no energy pumped into it. However, the market heat can sporadically turn into mechanical energy (at least in the damped stochastic oscillator model), just preventing the market from dying, even with damping and without exterior sources of energy. 

\subsection{Chemical market energy and other energy components}

When two stocks merge, the merger may release (or absorb) a lot of market energy. This is an example of what we call the chemical market energy, i.e. 
the energy related to "chemical" financial reactions. In mechanics, one often ignores this energy and other energy components, assuming them to be invariant
(and hence having no effect on the equation). For simplicity, in this
paper we will ignore these types of market energy. 

\subsection{Market energy, volume and volatility}

A lot of research papers on financial markets confirm the 
strong positive relation between trading volume and volatility,
and also the persistence (long memory) 
of volume and volatility in the market, see, e.g., 
\cite{LeBaron_AgentBased2006}. We propose here 
to explain these phenomena  by market energy.

Both volume and volatility are positively related
to market energy: the volume is roughly proportional to the kinetic market
energy, and so is the square of the volatility for some portion
of the kinetic energy.Hence one can conjecture that the volatility
is  highly correlated to the square root of the volume. 

In general,
it takes many trades (a lot of market energy) to move a market (especially
when the inertia is high), and the market movement is partly
reflected in volatility, and that's another way to say why volume and
volatility are highly positively correlated.

If one ignores external forces and dissipation then the market energy is
conserved. Due to this conservation principle, the energy is not 
conserved in general because the system is not closed, but it 
can't change very fast. That's why both volume and volatility has a memory.

Another explanation comes from the quasi-periodic nature of
the whole market (see Section \ref{section:N-Oscillator}):
 the market energy of each individual stock also changes
in a stochastically quasi-periodic manner.

\section{Individual stocks as damped stochastic oscillators}
\label{section:DampedOscillator}

\subsection{The model}

One of the simplest models for financial markets is the damped stochastic oscillator,
which can be used for the price movement of a single financial asset or commodity, 
such as gold, oil, SP500 index, or EUR/USD exchange rate, etc. 
In this model, the stochastic dynamical system is given by the stochastic vector field 
(see \cite{ZungThien_Stochastic2015}):

\begin{equation}
X = X_h +  D + \sigma B
\end{equation}
on the symplectic space $(\mathbb{R}^2, \omega= dx \wedge dy)$, 
where
\begin{equation}
D = - f(\sqrt{x^2 + y^2}).(x \partial x + y \partial y) 
\end{equation}
is the damping term,
\begin{equation}
B = \partial x  \circ \dfrac{d B^1_t}{dt} + 
\partial y  \circ \dfrac{d B^2_t}{dt}
\end{equation}
is the generator of a 2-dimensional Brownian motion (the noise), $\sigma$ is the volatility
coefficient (the amplitude of the noise), and
\begin{equation}
h = \dfrac{1}{2} (x^2 + y^2)
\end{equation}
is the market  energy of the asset. 
Here $x$ is the mispricing and $y$ is the
momentum. (Say $x = P - V$ or $x = \ln(P/V)$, where $P$ is the price and $V$ is the
fair value. We will assume that the price $P$ 
fluctuates a lot but the fair value $V$
varies very slowly over the time, so that most of the variation of $P$ is reflected in
the variation of $x$). The units (time, price, etc.) here are already normalized so
that $h$ takes the simplest form $h = \dfrac{1}{2} (x^2 + y^2)$. This model as a
$SO(2)$-symmetry which makes it into an integrable stochastic system, easy to investigate
(see \cite{ZungThien_Stochastic2015}).

The dissipative term $D$ in $X$ is due to market friction, e.g. trading fees, and has
the  energy-losing effect on the market. On the other hand, the noise
term $B$ has the energy-enhancing effect. These two effects cancel out
each other in a stochastic way. As a result, the expected market energy of a
single-stock market in this model does not die out (go to zero) or explode (go to
infinity) over time, but rather tends to a stable positive energy level. Similarly
to the mean-reverting  Ornstein-Uhlenbeck process, 
there is a stationary distribution density of
energy levels for the damped stochastic oscillator, which is concentrated around a
stable energy level. (Notice, however, that oscillators are not  mean-reverting:
they fo through the mean back and forth but do not ``converge'' to the mean. 
When they are at the mean they tend to go far away from the mean again 
if the energy is large enough).

The above simple oscillator model is a bit simplistic, because an individual 
stock is not a closed system like in the model, and the rest of the market can
affect it greatly. Nevertheless, it shows the stochastic quasi-periodic
nature and boom-bust cycles of real-world financial markets.

Let us look at two examples: gold prices and EUR/USD rates.

\begin{figure}[!ht]
\includegraphics[width=0.8\textwidth]{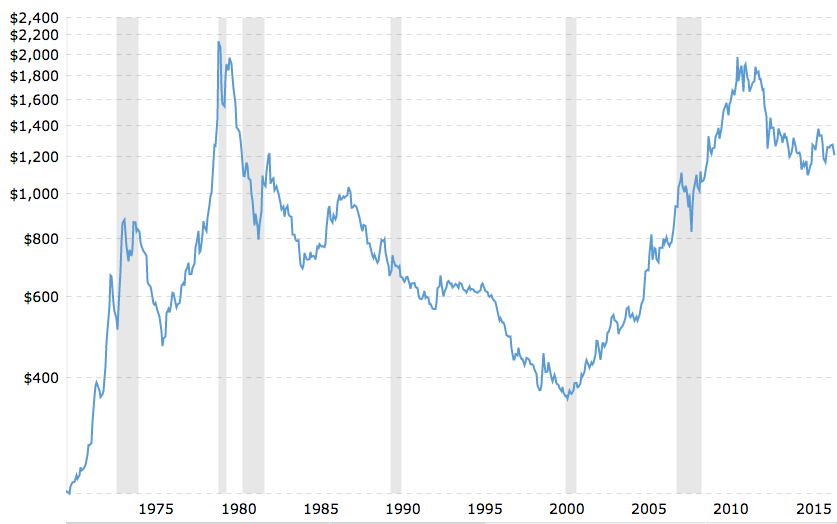}
\caption{Inflation-adjusted gold prices, from 1970 to 07/2017. 
Chart drawn by macrotrends.net}
\label{fig:Gold}
\end{figure}

\subsection{Gold price as a noisy oscillator}

Figure \ref{fig:Gold} is a inflation-adjusted
chart of historical gold prices from 1970 (not long before the US abandoned
the last gold peg 1 ounce = \$42.22) until the date of writing of this article 
(07/2017). No one knows for sure what is the fair value of gold, but one may argue 
that this inflation-adjusted fair value does not change much with time (according to 
macro-economic models), and that once in a while the price coincides with the fair value, i.e. the mispricing is 0 (according to the oscillator model). During the late 1970s and the 1980s, the price of gold is around 400 USD/ounce 
(or around 800-1000 USD in 2017's dollars, inflation-adjusted), 
so we may presume that the fair value is around those numbers at that time. 
When gold goes to 200 USD/ounce in 2000 (or under 400 USD in 2017's dollars), 
it becomes very underpriced according to the oscillator model, creating
a big potential speculation energy which results in a big upward movement later
on. In 2010s, the fair value of gold can be estimated at around 1000 USD/once
(inflation-adjusted). Of course, the fair value of gold doesn't have to say 
constant, but can change, due to the growth of world's economy, the growth of 
gold supply and other factors, but here for simplicity we assumed that didn't change much over the last 30-40 years. If it went up, say 30\% over the last 30-40 years, then the fair price of gold would be closer to 1300 than 1000 USD/ounce right now. When gold went above 1500 USD/ounce in 2011, it was already very probably 
overpriced, but it continued to move up due to positive momentum. Eventually this momentum died out, and what remained is a big potential speculation energy pointing 
to a big future downward movement. Sure enough, gold fell down from its peak 
of almost 2000 USD/ounce in 2011 to its current price of about 1200 USD/ ounce.

Notice also that during the period late 1980s and early 1990s, the price of gold
didn't move much, i.e. the speculation energy seems to die out during that period.
This loss of market energy in gold can't be explained in the single damped stochastic
oscillator model of the market, according to which the speculation energy will
(almost surely) never die out but will fluctuate around a certain energy level. 
However, it can be explained by using multi-body models of the market, where the speculation energy (or \textit{hot money} in financial jargon) 
can move from one component of the market to another.

Remark. Going back further in time, a free chart from goldchartsrus.com (not reproduced here)
shows that , since 1600 (4 centuries ago), inflation-adjusted gold prices 
oscillated around 450USD/ounce a great number of times and rarely shot up above 1000USD,. So if we assume that the inflation data over the last 4 centuries is correct (which is a big assumption, maybe not true), then gold is right now still more 
expensive than during centuries ago.  
   
\subsection{EUR/USD exchange rates}

Figure \ref{fig:EUR-USD-PPP} is a chart 
chart of historical EUR/USD exchange rates over the past 20 years. Since the USA
and the Eurozone have slightly different inflation rates over those years,
it may be better to divide the exchange rate by the PPP (purchasing power parity) 
EUR/USD curve, which serves as a kind of ``fair price'' curve, and which 
is also given in Figure \ref{fig:EUR-USD-PPP}. 

\begin{figure}[!ht]
\includegraphics[width=0.8\textwidth]{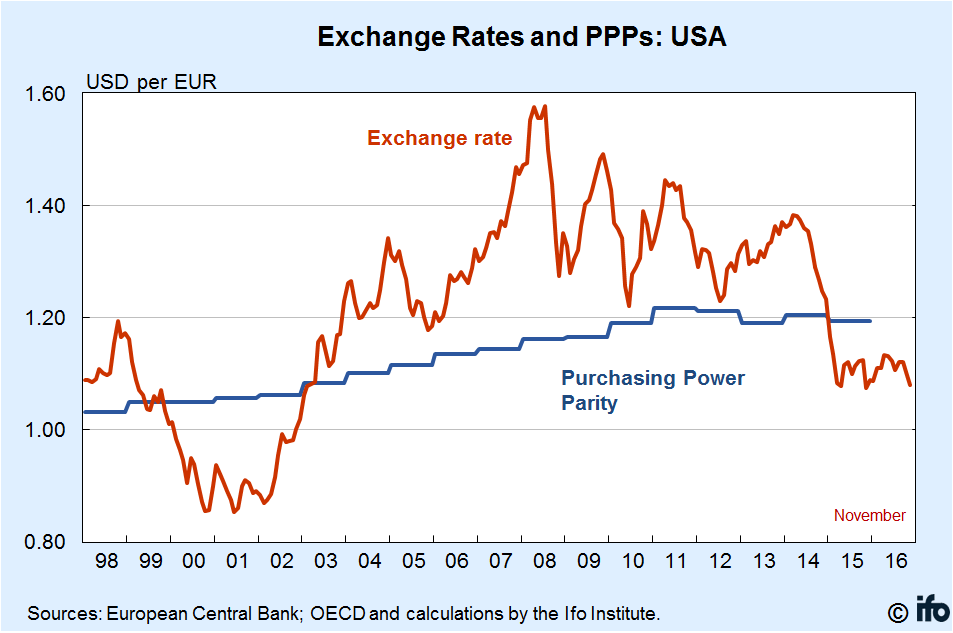}
\caption{EUR/USD exchange rate and PPP from 1997 to 2017. 
Source: https://www.cesifo-group.de}
\label{fig:EUR-USD-PPP}
\end{figure}

The quasi-periodic oscillating nature of the EUR/USD exchange rates around 
the PPP is clear from the Figure \ref{fig:EUR-USD-PPP}. Nevertheless, one may notice
that the potential energy function looks more like a double-well function (with an
optimist and a pessimist perceived fair value for EUR/USD) than a single well. 

\section{Market patterns}
\label{section:TA}

\textit{Technical analysis}, i.e. the search for market patterns, 
is used by a great number of market players with various degrees of success
(see, e.g., \cite{AGS_TA2000,BLL_TA1992,CaCo_Momentum1995,
CaDe_NonlinearityMarkets2011,LuiChong_TA2013}), and is 
at odds with the \textit{efficient market hypothesis} (see, e.g., \cite{Fama-EMH1970,LoMa_NonRandom2001,Samuelson_Random1965}). For that matter, most agent-based models are also at odds with this hypothesis. 

In this section, we want to give an explanation of some simple
and easily recognizable market patterns by using market energy and
second order models. Namely, we will discuss three market patterns:
U-shaped vs. V-shaped reversals; 
resistance breaking; and market aftershocks. 

\subsection{U-shaped versus V-shaped reversals}

The financial price reversals are often divided into 2 main types: U-shaped
and V-shaped. The difference between U-shaped ad V-shaped reversals is
in the kinetic energy: at the point of a U-shaped reversal, the kinetic energy goes to
0, i.e., the momentum of the stock (end hence the kinetic energy)
dies out before reversing, like a ball going up
and then makes a U-turn and falls down on its own weight.
On the contrary, in a V-shaped reversal situation, when the stock hits a hard ``wall'' 
(strong resistance), the kinetic energy remains positive, the absolute value of the 
momentum does not change much, it's just the direction which changes, similarly
to an elastic \textit{bouncing ball} when hitting the wall.

An important physical property of physical objects which can bounce back strongly 
when hitting a wall is their elasticity. So apparently,
the market is also elastic when it makes V-shaped reversals, and this elasticity 
might be explained by the prevalence of active market swingers, i.e. 
active traders (or trading strategies) who 
switch sides easily when the stock hits a resistance.

\begin{figure}[!ht]
\includegraphics[width=0.8\textwidth]{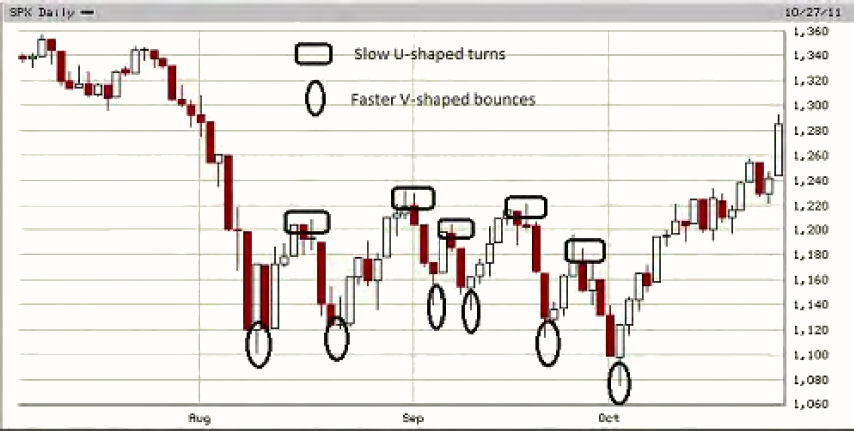}
\caption{SP500 Index barchart from 07-2011 to 10-2011}
\label{fig:SP500}
\end{figure}

For example, 
Figure \ref{fig:SP500} is a daily chart of the SP500 index for the period from
07/2011 to  11/2011. Notice how it also moved like a bouncing ball during the
months 08/2011 -- 09/2011: every time it falls down to a level near 1100 it makes
a V-shaped (fast) reversal, but when it goes up to around 1200 it makes a U-shaped (slow) reversal.

Notice also that V-shaped patterns can be seen more often in short-term movements, rather than long-term movements, of a financial asset price. That is because the  ``walls'' can often be set up by
the ``houses'', who are strong enough to control the price of
a stock short-term, but not over the long term.

\subsection{Resistance breaking}

A market resistance may sometimes be analogous to a 
dike which prevents water waves from overflowing (or a wall
as in the V-shaped market reversal pattern discussed above). 
But when the waves are strong enough to break the dike, 
i.e. the market momentum is strong enough to break the resistance, 
there will be a flood, i.e. a large market move once the resistance is broken.

\begin{figure}[!ht]
\includegraphics[width=0.8\textwidth]{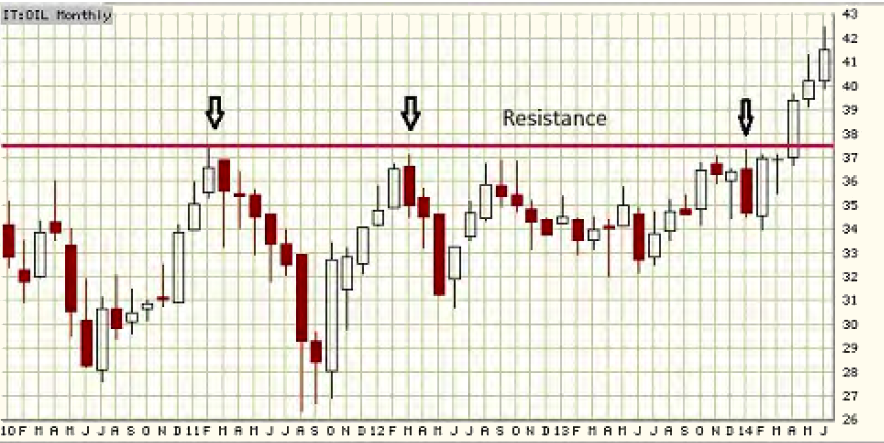}
\caption{Resistance breaking of Lyxor ETF STOXX Europe 600 Oil Gas -
chart between 2010-2014}
\label{fig:Oil2010-14}
\end{figure}

Figure \ref{fig:Oil2010-14} is typical example of resistance
breaking. Notice how the price bounced back 
the first times it hit the resistance, and then finally the resistance
is broken due to strong market energy.

From the point of view of market energy, a wall is a sharp spike in the
potential energy function near some point of the price variable. In order
to go over this potential wall to the other side of the price region,
the market needs a lot of energy. Maybe during the first attempts at coming
close to the wall, there is not enough market energy to go over it, so the
price falls back (potential energy changing into kinetic energy). But with
some additional from the outside (for example, the whole stock market
is moving in some direction, giving additional kinetic energy to the stock), the potential wall is finally overcome, and after that it often happens that the stock price moves very fast due to high potential energy
turning into kinetic energy after the wall.

\subsection{Market aftershocks}

After a big earthquake hits some area,
there are often after shocks, which are less strong but can also 
be quite violent. It is partly because a lot of energy is still there
and cannot dissipate too quickly. The same happens in the stock market:
the level of market energy created after a shock is high, and this high energy results in big (usually oscillating) aftershock moves.

\begin{figure}[!ht]
\includegraphics[width=0.8\textwidth]{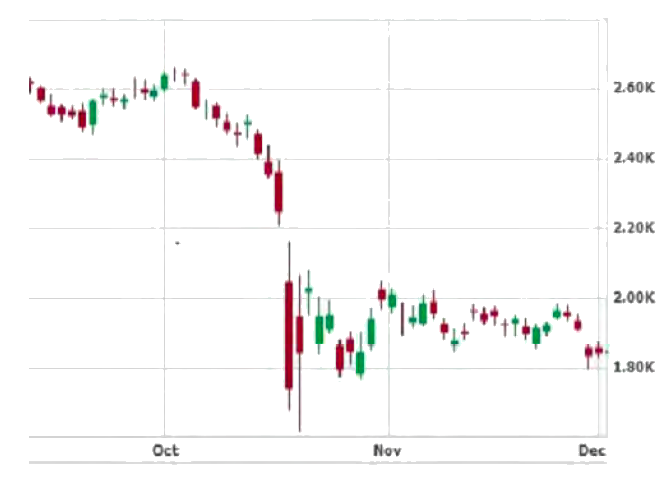}
\caption{Dow Jones Industrial Average crash 10/1987}
\label{fig:DJIA1987}
\end{figure}

Figure \ref{fig:DJIA1987} is an example of a big stock market 
shock in 1987, together with large aftershock movements. 

Notice that, just as earthquakes are often localized,
market shocks happen much more easily in individual 
stocks than for the whole market: a little market energy may already
be enough for a small stock to make a huge move. 

\section{The market as a constrained n-oscillator}
\label{section:N-Oscillator}

\subsection{The linear deterministic model} \hfill

In this model, we consider the total market which
contains every economical asset, i.e. it represents the whole economy.
The assets are divided into $n$ asset classes $A_1,\hdots,A_n$, 
for example : energy, real estate, food, transport, communication, etc. 
The total price of each asset class is denoted by $P_i = p(A_i)$. Then
$P = \sum P_i$ is the total net worth of the whole economy, and 
we call 
\begin{equation}
R_j = \frac{P_j}{\sum P_i}
\end{equation}
the \textit{\textbf{relative price}} of the asset class $A_j$, so that
$\sum R_j = 1$.

We assume that class $A_i$ has a relative fair value $v_i$ in the economy 
($\sum v_i = 1$), which  varies very slowly with the time. 
For example, people will pay only a certain percentage of their income 
for telecommunication needs, and therefore fast technological advances 
don't make this sector occupy a much larger share of the whole economy, 
but make the prices per unit drop fast instead. We will be interested in 
the asset mispricings
\begin{equation}
x_i = R_i -  v_i.
\end{equation}
Similarly to the oscillator model for a single stock, we will assume that the
market energy has the form
\begin{equation} \label{eqn:Energy1}
E = \dfrac{1}{2} \sum_1^n a_i x_i^2 +  \dfrac{1}{2} \sum_1^n b_i \dot{x}_i^2.
\end{equation}
Here $\frac{1}{2} \sum a_i x_i^2$
is the potential energy, $\dfrac{1}{2} \sum_1^n b_i \dot{x}_i^2$ is the kinetic energy,
and $a_i, b_i > 0$ are constant asset-specific coefficients. 
So we get a Hamiltonian system with the energy function E given by Formula
\eqref{eqn:Energy1} and a linear constraint
\begin{equation}\label{eqn:LinearConstraint1}
\sum x_i = 0. 
\end{equation}
Since the constraint is holonomic, this is a
Hamiltonian system with $n-1$ degrees of freedom.

In order to write down the equation of motion one can for example eliminate
one of the variables (say by putting $x_n = - \sum_{i=1}^{n-1} x_i$)  
and consider it as a system
on $T^*\mathbb{R}^{n-1}$. Equivalently, one can use the Lagrangian multiplier 
method as follows:

The Lagrangian action function is:
\begin{equation}
L = \frac{1}{2} \sum a_i x_i^2 - \frac{1}{2} \sum b_i \dot{x}_i^2
\label{(4.8)}
\end{equation}

The equation is
$\dfrac{\delta L}{\delta x_i} = \lambda \dfrac{\partial f}{\partial x_i}
\quad \forall\ i = 1,\hdots,n,$
where $f(x) =  \sum x_i$ is the constraint function, and  $\lambda$
is the Lagrangian multiplier to be determined.
Since
$
\dfrac{\delta L}{\delta x_i} =\dfrac{\partial L}{\partial x_i} 
- \dfrac{d}{dt}\dfrac{\partial L}{\partial \dot{x}_i} 
= a_i x_i + b_i \ddot{x}_i
$
and
$
\dfrac{\partial f}{\partial x_i} = 1,
$
we get the system of equations:
\begin{equation} \label{eqn:motion16}
a_i x_i + b_i \ddot{x}_i = \lambda \quad  \forall\  i=1,\hdots, n.
\end{equation}

Since $\sum x_i = 0$ implies
$\sum \ddot{x}_i = 0$, we get the following equation for $\lambda$:
$ \displaystyle \sum \dfrac{\lambda}{b_i} = \sum \dfrac{a_ix_i}{b_i}$,
which implies that
$\lambda = 
\left( \sum \dfrac{a_ix_i}{b_i} \right) / \left( \sum \dfrac{1}{b_i} \right).$
Thus the equation of motion \eqref{eqn:motion16} is a system of linear 
differential equations with
constant coefficients
\begin{equation}\label{eqn:ddotxi_withlinearconstraint}
\ddot{x}_i = \left(\sum \dfrac{a_jx_j}{b_j}\right)/
\left(b_i\sum \dfrac{1}{b_j}\right)  - \dfrac{a_ix_i}{b_i}
\end{equation}
and with the constraint $\sum x_i = 0$.

\begin{proposition}
By a linear transformation
\begin{equation}
x_i = \sum c_{ij} z_j,
\end{equation}
where $(c_{ij})$ is an appropriate constant 
$n \times (n-1)$ matrix of rank $n-1$, the system \eqref{eqn:ddotxi_withlinearconstraint} 
with constraint \eqref{eqn:LinearConstraint1} 
becomes a system of $n-1$ uncoupled harmonic oscillators:
\begin{equation}
\ddot{z}_i = - \lambda_i^2 z_i.
\end{equation}
with constants $\lambda_1,\hdots, \lambda_{n-1} > 0$ (called \textbf{\textit{eigenvalues}} or
\textbf{\textit{normal modes}} of \eqref{eqn:ddotxi_withlinearconstraint}).
\end{proposition}

The proof is a simple exercise in classical mechanics \cite{Arnold-Mechanics1978}: 
when written as a Hamiltonian system on 
$T^*\mathbb{R}^{n-1} = \mathbb{R}^{2(n-1)}$, we have a positive definite
quadratic Hamiltonian function, and any such function can be written as
$\dfrac{1}{2}\sum \lambda_i (z_i^2 + w_i^2)$
in some linear canonical coordinate system $(z_i^2,w_i)$ on $\mathbb{R}^{2(n-1)}$.

The general solution of our model has the form:
\begin{equation}
x_i = \sum^{n-1}_{j=1}
c_{ij} \sin(\lambda_j t + d_{ij})
\end{equation}
with appropriate coefficients $c_{ij}$ and $d_{ij}$ 
(so that the constraint $\sum x_i = 0$ is satisfied). 
Thus, the mispricing of each asset  is 
a quasi-period function with periods $\lambda_1,\hdots,\lambda_{n-1}$. 
Notice that all the asset classes share the same periods.

We will call the above simple model the (linear deterministic) 
\textit{\textbf{constrained $n$-oscillator model}} of the market.

\begin{remark}{\rm
In some physics textbooks (e.g., \cite{Bajaj_Waves2006}) 
one can find a so called coupled
$n$-oscillator model, which explains the waves in materials and which consists
of a chain of masses connected to each other by springs. Our model is similar to,
but different from, this coupled $n$-oscillator model, because the kinetic energy in
our model is different from the kinetic energy is not the same.
}
\end{remark}

\begin{figure}[!ht]
\includegraphics[width=0.6\textwidth]{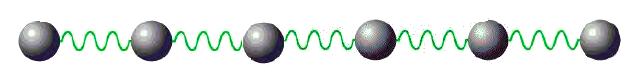}
\caption{A coupled n-oscillator model in physics, which is not the same as
our constrained n-oscillator model}
\end{figure}

\subsection{Frequencies of the system} \hfill

In the linear model with the energy function
$E = \dfrac{1}{2} \sum a_i x_i^2 +  \dfrac{1}{2} \sum b_i \dot{x}_i^2$
and the constraint $\sum x_i = 0$, we will call
\begin{equation}
E_i =  \dfrac{1}{2}  a_i x_i^2 +  \dfrac{1}{2}  \dot{x}_i^2
\end{equation}
the \textbf{\textit{energy of the $i$-th component}}, 
and the number $ \gamma_i = \sqrt{a_i/b_i}$
the \textbf{\textit{proper frequency}} of the $i$-th component. 
If there were no constraint then $x_i(t)$ would be
a periodic function in time $t$  with period $\dfrac{2\pi}{\gamma_i}$.

The following proposition shows the relationship between 
the frequencies of the linear
constrained $n$-oscillator and the proper frequencies of its components.

\begin{proposition}
Assume that the proper frequencies 
$\gamma_1 = \sqrt{a_1/b_1},\hdots, \gamma_n = \sqrt{a_n/b_n}$
of the components of the above linear constrained $n$-oscillator 
are ordered in an increasing way:
\begin{equation}
\gamma_1 \leq \gamma_2 \leq \hdots \leq \gamma_n.
\end{equation}
Then the linear constrained $n$-oscillator is equivalent to a system of $(n-1)$ 
free (uncoupled) harmonic oscillators whose frequencies 
$\lambda_1,\hdots,\lambda_{n-1}$ 
satisfy the inequality
\begin{equation} \label{eqn:GC_Inequalities}
\gamma_1 \leq \lambda_1 \leq \gamma_2  \leq  \lambda_2 \leq \hdots \leq \gamma_{n-1} 
\leq\lambda_{n-1} \leq \gamma_n.
\end{equation}
Conversely, if $\gamma_i$ and $\lambda_i$ are arbitrary positive numbers which 
satisfy \eqref{eqn:GC_Inequalities} in the strict sense (i.e., there is no equality), 
then there exist positive numbers $a_i, b_i$ such that 
$\gamma_i = \sqrt{a_i/b_i}$
and the frequencies of the above constrained n-oscillator 
are $\lambda_1,\hdots,\lambda_{n-1}$.
\end{proposition}

\begin{proof}
With a linear change of the coordinates $y_i = \sqrt{a_i}x_i$, 
we can write
\begin{equation}
E = \dfrac{1}{2} \sum y_i^2 +  \dfrac{1}{2} \sum \dfrac{\dot{y}_i^2}{\gamma_i^2}
\end{equation}
with the constraint $\sum \alpha_i y_i = 0$, where $\alpha_i = 1/ \sqrt{a_i}$. 
The corresponding unconstrained equation is:
\begin{equation}
y + \Gamma.\ddot{y} = 0,
\end{equation}
where  
$\Gamma = \text{diag} (1/\gamma_i^2)$ is the diagonal matrix whose entries are
$1/\gamma_i^2$. The true equation of motion, taking into account the constraint, is:
\begin{equation}
y + \Gamma. \ddot{y} \in \mathbb{R}. (\alpha_1,\hdots,\alpha_n)^T,
\end{equation}
where $y = (y_1,\hdots,y_n)^T$.
(The above equation is written in the form of an inclusion,
which means that $y + \ddot{y}$ is collinear to
$(\alpha_1,\hdots,\alpha_n)^T$; $T$ means the transpose).

Let $O$ be an orthogonal matrix such that
\begin{equation}
O. (\alpha_1,\hdots,\alpha_{n-1},\alpha_n)^T = (0,\hdots,0, \norm\alpha )^T,
\end{equation}
where $\norm\alpha = \sqrt{\sum \alpha_i^2}$.  
Then above equation is equivalent to:
\begin{equation}
Oy + O\Gamma \ddot{y} \in \mathbb{R}. (0,\hdots,0, 1)^T.
\end{equation}
Denote
\begin{equation}
z = (z_1,\hdots, z_n) = Oy.
\end{equation}
Then the constraint $\langle y, \alpha \rangle = 0$ is equivalent to 
$\langle z, O\alpha \rangle = \langle Oy, O\alpha \rangle = 
\langle y, \alpha \rangle = 0,$  i.e. $z_n = 0$, i.e. we can forget about $z_n$ 
and look only at the 
variables $z_1, \hdots, z_{n-1}$.

Denote by $A$ the left-top $(n-1) \times (n-1)$ minor of the positive symmetric matrix
$O\Gamma O^{-1}$, then the system is equivalent to
\begin{equation}
z + A \ddot{z} = 0.
\end{equation}

It is well-known that the eigenvalues of a symmetric matrix
$O\Gamma O^{-1}$ (which are $\dfrac{1}{\gamma_1^2},\hdots, \dfrac{1}{\gamma_{n}^2}$)
and the eigenvalues of its left-top minor $A$ 
(which are $\dfrac{1}{\lambda_1^2},\hdots, \dfrac{1}{\lambda_{n-1}^2}$)
satisfy Inequality \eqref{eqn:GC_Inequalities}, which is part of the so-called 
\textit{Gelfand-Ceitlin triangle of inequalities}, see, e.g., \cite{GS_GC1983}.
\end{proof}

\begin{remark}{\rm
 The component energy functions 
 $E_i =  \dfrac{1}{2}  a_i x_i^2 +  \dfrac{1}{2}  \dot{x}_i^2$
are not first
integrals of the constrained n-oscillator model, i.e. they also change with time
(in a quasi-periodic way). Thus we see a speculation \textit{energy transfer among the
components} of the market in this model.
}
\end{remark}

\subsection{Market sectors: Components having the same proper frequencies} \hfill

In the generic case, the frequencies $\lambda_1,\hdots,\lambda_{n-1}$ 
of the linear unconstrained $n$-oscillator are incommensurable, 
and the regular minimal invariant tori of the system in the phase space are 
of full dimension $n-1$. However, there are some special
cases when the minimal invariant tori are of dimension less than $n-1$. 
One particular case is when there are some components whose proper frequencies 
are the same.

Assume, for example:
\begin{equation}
\gamma_{p+1} = \gamma_{p+2} = \hdots = \gamma_{p+k}, \quad\ p \geq 0, k \geq 2.
\end{equation}
Then, according to Inequality \eqref{eqn:GC_Inequalities}, we also have:
\begin{equation}
\gamma_{p+1} = \lambda_{p+1} = \gamma_{p+2} = \hdots 
= \lambda_{p+k-1} = \gamma_{p+k}, 
\end{equation}
i.e., the multiplicity of the frequency $\lambda_{p+1}$ in the linear 
unconstrained $n$-oscillator
is (at least) $k -1$, and this frequency coincides with the proper frequency of $k$
components of the system.

By putting $\hat{x}_1 = x_1,\hdots,\hat{x}_p = x_p,$
$\hat{x}_{p+2} = x_{p+k+1},\hdots,\hat{x}_{n-k+1} = x_n$, and
\begin{equation}
\hat{x}_{p+1} = x_{p+1} + \hdots + x_{p+k}, 
\end{equation}
we can make a reduction of the system in this
case, reducing the number of components from $n$ to $n - k + 1$ (and killing the
frequency $\lambda_{p+1}$ along the way, 
by ``averaging out'' with respect to that frequency).
This procedure corresponds to the practice of regrouping many similar components
into a sector in the market.

The coefficients associated to the sector $(x_{p+1},\dots,x_{p+k})$ are:
\begin{equation}
\hat{a}_{p+1} = \dfrac{1}{\sum_{i=1}^k \frac{1}{a_{p+i}}}, \; 
\hat{b}_{p+1} = \gamma_{p+1}^2 \hat{a}_{p+1}.
\end{equation}
(The coefficients for the other components remain the same: $\hat{a}_i = a_i$
and $\hat{b}_i = b_i$ for $i \neq p + 1,\dots,p + k)$:
The speculation energy of the sector
\begin{equation}
E_\text{sector} = \sum_{i=1}^k E_{p+i} = 
\sum_{i=1}^k \dfrac{1}{2} (a_{p+i} x_{p+i}^2 + b_{p+i} \dot{x}_{p+i}^2)
\end{equation}
is decomposed into the sum of 2 parts: the external energy (vis a vis the market)
and the internal speculation energy (which accounts for the internal movements
in the sector):
\begin{equation}
\hat{E}_{p+1} = E_\text{external}  
= \dfrac{1}{2} (\hat{a}_{p+1} \hat{x}_{p+1}^2 + \hat{b}_{p+1} \dot{\hat{x}}_{p+1}^2)
= \dfrac{1}{2} \dfrac{ (\sum_{i=1}^k \hat{x}_{p+i})^2}{\sum_{i=1}^k \frac{1}{a_{p+i}}}
+ \dfrac{1}{2} \dfrac{ \gamma_{p+1}^2(\sum_{i=1}^k \dot{\hat{x}}_{p+i})^2}{\sum_{i=1}^k \frac{1}{a_{p+i}}}
\end{equation}
and
\begin{multline}
E_\text{internal} = E_\text{sector} - E_\text{external} \\
= \dfrac{1}{2}\left[  \sum_{i=1}^k \dfrac{1}{2} a_{p+i} x_{p+i}^2 
- \dfrac{1}{2} \dfrac{ (\sum_{i=1}^k \hat{x}_{p+i})^2}{\sum_{i=1}^k \frac{1}{a_{p+i}}} \right] 
+ \dfrac{\gamma_{p+1}^2)}{2}
\left[  \sum_{i=1}^k \dfrac{1}{2} b_{p+i} \dot{x}_{p+i}^2 
- \dfrac{1}{2} \dfrac{ (\sum_{i=1}^k \dot{\hat{x}}_{p+i})^2}{\sum_{i=1}^k \frac{1}{a_{p+i}}} \right] 
\end{multline}
(The energy of the other components remains the same).

Remark the natural fact that $E_\text{internal} \geq 0$, 
and this inequality can be seen as a particular case of the Cauchy-Schwartz inequality
\begin{equation}
\left( \sum a_{p+i} x_{p+i}^2 \right) \left( \sum \frac{1}{a_{p+i}} \right)
\geq \left( \sum \sqrt{a_{p+i} x_{p+i}^2} \sqrt{\frac{1}{a_{p+i}}} \right)^2
= \left( \sum x_{p+i} \right)^2 .
\end{equation}

The internal movement (among the components of the sector, but does not affect
the total sector mispricing $z_{p+1} = \sum^k_{i=1} x_{p+i})$ 
is governed by the internal energy
function $E_\text{internal}$. This movement is periodic of period 
$\dfrac{2\pi}{\lambda_{p+1}}$ (frequency $= \lambda_{p+1}$)
and is isomorphic to a synchronous $(k-1)$-dimensional harmonic oscillator (i.e.
Hamiltonian system with Hamiltonian function 
$\displaystyle  h = \dfrac{\lambda_{p+1}}{2} \sum_1^{k-1} (p^2_i + q^2_i))$ 
on the symplectic space 
$(R^{2(k-1)}, \omega = \sum^{k-1}_{i=1} dp_i \wedge dq_i)$). 
This internal movement commutes
with the external movement of the market, which now has $n-k+1$ components
$z_1,\dots,z_{n-k+1}$ instead of $n$ components and the new speculation energy function
\begin{equation}
\hat{E} =
\sum^{n-k+1}_{j=1}
\hat{E}_j = E_1 + \dots + E_p + E_\text{external} + E_{p+k+1} + \dots + E_n. 
\end{equation}

We can reduce the system, from a constrained $n$-oscillator to a constrained 
$(n-k + 1)$-oscillator, by forgetting about the internal movement in the 
sector consisting of $k$ components $x_{p+1},\hdots,x_{p+k}$ 
and considering the whole sector as just
one component $\hat{x}_{p+1} = x_{p+1} + \dots + x_{p+k}$.

\subsection{The stochastic model} \hfill

Our stochastic constrained $n$-oscillator model of the market will be a perturbation
of the deterministic linear constrained $n$-oscillator model, which is a proper
integrable Hamiltonian systems with $n-1$ degrees of freedom. Under a nonlinear
perturbation, an integrable systems is no longer integrable in general and may exhibit
chaotic behavior. Nevertheless, the KAM (Kolmogorov-Arnold-Moser) theory
with Nekhoroshev's exponential time stability theory 
(see, e.g.,  \cite{Meyers_Dynamics2011}) 
say that if the perturbation is deterministic and small and the system is 
non-resonant, then
the most solutions of the perturbed system are still quasi-periodic, 
at least for a very long period of time. 

When stochastic terms are added, the situation
becomes more complicated. There are elements of KAM theory in the stochastic
case (for example, the averaging method with respect to a torus action, see, e.g.,
\cite{FrWe-Random2012,SK_StochasticAveraging2015} and references therein), 
but as far as we know, a full KAM theory for SDS does not exist yet. 
Nevertheless, we will assume that most solutions of a reasonable 
stochastic perturbation of an integrable Hamiltonian system will look 
similar to solutions of an integrable stochastic  dynamical systems, at least for a
very long period of time. For practical purposes, here we will be interested only
in such solutions. So we will look only at stochastic models which are integrable
in the sense of \cite{ZungThien_Stochastic2015}, or even more restrictively, which 
are invariant with respect to a torus action of half the dimension of the phase space, 
similarly to classical integrable Hamiltonian systems and their Liouville 
torus actions (see \cite{Zung-AA2017}).

In the deterministic linear constrained $n$-oscillator model, 
the general solution has
the form:
\begin{equation}
x_i(t) = \sum^{n-1}_{j=1} c_{ij}z_j(t); \quad i = 1,\hdots,n,
\end{equation}
with
\begin{equation}
z_j(t) = r_j \sin(\lambda_j t + \theta_j); \quad j = 1,\hdots,n-1,
\end{equation}
where $(c_{ij})$ is a constant matrix of linear transformation, 
$r_j > 0$ ($j = 1,\hdots,n-1$)
are action coordinates which do not depend on time 
(they are first integrals of the
system), and $\lambda_j t + \theta_j$ are angle coordinates which more at constant 
frequencies $\lambda_j$ .
The numbers $(r_j,\theta_j)$ are initial data in the action-angle coordinate system.

In our simplest stochastic model, we will use the same linear transformation 
matrix $(c_{ij})$ to write
$x_i(t) = \sum^{k-1}_{j=1} c_{ij}z_j(t)$ for $i = 1,\hdots,n$, 
and will assume that each $z_j$ behaves like a
damped stochastic oscillator. The general solution for $z_j$ has the form
\begin{equation}
z_j(t) = r_j(t) \sin(\lambda_j t + \theta_j + S_j(t)),
\end{equation}
where $r_j(t)$ is no longer constant in $t$ but satisfies a stochastic 
differential equation of the form
\begin{equation}
dr_j(t) = \left(  \dfrac{1}{r_j(t)} - f(r_j(t))\right) dt + dB^j_t 
\end{equation}
(like the one obtained for the 1-degree-of-freedom damped stochastic oscillator
in polar action-angle coordinates, see \cite{ZungThien_Stochastic2015}) 
where $S_j(t)$ is a martingale Ito process whose volatility is inverse 
proportional to $r_j(t)$: 
$dS_j(t) =\dfrac{1}{r_j(t)}dW^j_t$. (Here $B^j_t$ and $W^j_t$
are independent Wiener processes). 

We will call the process satisfied by each
$r_j$ a \textbf{\textit{positive bell-shaped process}}, 
in view of the shape of its stationary density
function.

In summary, our stochastic model is as follows:
\begin{equation}
x_i(t) = \sum^{n-1}_{j=1} 
c_{ij}r_j(t) \sin(\lambda_j t + \theta_j + S_j(t)), \quad\ i = 1,\hdots,n,
\end{equation}
where:
\begin{itemize}
\item $x_i(t)$ is the mispricing of $i$-th component at time $t$,
\item $(c_{ij})$ is a constant matrix of linear transformation,
\item $r_j(t)$ are independent positive bell-shaped processes $(j = 1,\hdots,n)$,
\item $\lambda_j > 0$ are frequencies,
\item $\theta_j$ are initial angular values,
\item $S_j(t)$ are independent martingale Ito processes whose 
volatilities are $\dfrac{1}{r_j(t)}$ respectively,
i.e., $\displaystyle dS_j(t) = \dfrac{1}{r_j(t)} dW^j_t \quad
(j = 1,\hdots, n-1).$
\end{itemize}

This model has the following features, which are compatible with observations 
in the real-world financial markets: 

- The whole system is integrable quasiperiodic in stochastic sense, and goes
through boom-bust cycles. 

- Every component (asset) has the same set of periods, but different periods have
different relative importances (the coefficients $c_{ij}$) for different assets.

- At any given time, different momenta corresponding to 
different periods may be of the same sign, or they may be of opposite signs 
(i.e. they are counter-trend to each other) resulting in a complicated zig-zag
movement of the price (even before the noises).

- The assets can be regrouped into sectors according to their proper frequencies.
Each sector has its external motion (vis a vis the whole market) and internal
motion (change of relative weights of the stocks in the sector - this internal
change is periodic in the stochastic sense and has its own period).

\section*{Acknowledgement}

This paper was written during  my stay at the School of 
Mathematical Sciences, Shanghai Jiao Tong University, as a visiting 
professor. I would like to thank Shanghai Jiao Tong University, the 
colleagues at the School of Mathematics of this university, and 
especially Tudor Ratiu, Xiang Zhang, Jianshu Li, 
and Jie Hu for the invitation, hospitality and excellent working 
conditions. 

I told some of the ideas of this paper to my former student N.T. Thien,
who included them in his thesis in 2014. We were supposed to develop them
together, but he found a job in a bank right after his thesis and didn't
have time to do more research. Nevertheless, we wrote a paper together 
\cite{ZungThien_Stochastic2015} which provides some mathematical 
background for this paper.


\begin{thebibliography}{10}

\bibitem{AGS_TA2000}
J. V. Andersen, S. Gluzman and D. Sornette, 
\textit{Fundamental Framework for Technical Analysis}, 
European Physical Journal B 14, 579-601 (2000).

\bibitem{Arnold-Mechanics1978}
V.I. Arnold, Mathematical method of classical mechanics, 
Springer-Verlag, 1978.

\bibitem{Bachelier1900}
L. Bachelier, \textit{Théorie de la Spéculation}, 
Ann. Sci. Ecole Norm. Sup. 17 (1990), 21-86.

\bibitem{Bajaj_Waves2006}
N.K. Bajaj, The physics of waves and oscillations, 
McGraw-Hill 1984, (20th reprint, 2006).

\bibitem{Bismut-Aleatoire1981}
J.M. Bismut, Mecanique Aléatoire. Lecture Notes in Mathematics, 
Volume 866, Springer-Verlag 1981.

\bibitem{BS1973}
F. Black and M. Scholes, 
\textit{The Pricing of Options and Corporate Liabilities}, Journal
of Political Economy 81 (1973), No. 3, 637-654.

\bibitem{BouchaudCont_LangevinCrash1998}
J.-P. Bouchaud and R. Cont,
\textit{A Langevin approach to stock market fluctuations and crashes},
Eur. Phys. J. B 6 (1998), 543-550.


\bibitem{BLL_TA1992}
W. Brock, J. Lakonishok, B. Lebaron, \textit{Simple 
Technical Trading Rules and the Stochastic Properties of Stock Returns},
The Journal of Finance. 47 (1992), No. 5, 1731–1764.

\bibitem{CaCo_Momentum1995}
G. Caginalp and G. Constantine, 
\textit{Statistical inference and modeling of momentum in stock prices},
Applied Mathematical Finance 2 (1995), 225-242.

\bibitem{CaDe_NonlinearityMarkets2011}
 G. Caginalp and M. DeSantis, \textit{Nonlinearity in the dynamics of 
 financial markets}, Nonlinear Analysis: Real World Applications, 12 (2011),
 No. 2, 1140-1151.

\bibitem{CMZ_Prototype1997}
G. Caldarelli, M. Marsili, Y.-C. Zhang,
\textit{A Prototype Model of Stock Exchange},
Europhys Lett, 40 (1997), 479–484.


\bibitem{ContBouchaud_Herd2000}
R. Cont and J.-Ph. Bouchaud, 
\textit{Herd behavior and aggregate fluctuations in financial markets},
Macroeconomic Dynamics, 4 (2000), 170–196. 

\bibitem{DayHuang_BullBear1990}
R. Day and W. Huang, \textit{Bulls, bears and market sheep}, 
Journal of Economic Behavior and Organization, 14 (1990), 299-329.



\bibitem{Fama-EMH1970}
E. Fama, \textit{Efficient Capital Markets: A Review of Theory 
and Empirical Work}, Journal of Finance. 25 (1070), No. 2, 383–417. 

\bibitem{Farmer_AgentBased2001}
J.D. Farmer, \textit{Toward Agent-Based Models for Investment},
 Association for Investment Management and Research, 2001.

\bibitem{Farmer_Challenge2011-2014} 
J.D. Farmer, slides of the talks: 
\textit{The challenge of building agent-based models of
the economy}, European Central Bank, Frankfurt, June 10, 2011;
and: \textit{The Challenge of Agent-based Modeling in Economics},
ESRC Conference on Diversity in Macroeconomics,
University of Essex, Feb. 24, 2014.

\bibitem{FLPPS_AgentStochastic2012}
L. Feng, B. Lia,B. Podobnik, T. Preis and H.E. Stanley,
\textit{Linking agent-based models and stochastic models of financial markets}, 
PNAS, Vol. 109 (2012), no. 22.


\bibitem{FrWe-Random2012}
M. Freidlin, A.D. Wentzell, 
Random perturbations of dynamical systems.
Third edition. Grundlehren der Mathematischen Wissenschaften, 
260.  Springer, Heidelberg, 2012. xxviii+458 pp. 



\bibitem{GiaBou_AgentBased203}
I. Giardina and J.-P. Bouchaud,
\textit{Bubbles, crashes and intermittency in agent based 
market models}, The European Physical Journal B, 31 (2003),
Issue 3, pp 421–437.

\bibitem{GS_GC1983}
V. Guillemin and S. Sternberg, 
\textit{The Gel'fand-Cetlin system and quantization of the complex flag manifolds}, 
Journal of Functional Analysis, 52 (1983), No. 1, 106–128.

\bibitem{Gitterman-NoisyOscillator2005}
M. Gitterman, The Noisy Oscillator: 
The First Hundred Years, from Einstein Until Now, 
World Scientific Publishing Company, 2005. 

%
\bibitem{GK_StochasticKAM2014}
H. Guan, S.Kuksin, 
\textit{The KdV equation under periodic boundary conditions and its perturbations},
J. Nonlinearity, V 27 (2014), No. 9, pp. 61--88.

%

\bibitem{Hsu_StochasticAnalysis2002}
E.P. Hsu,  Stochastic Analysis on Manifolds. 
Graduate Studies in Mathematics, American Mathematical Society, 2002.

\bibitem{KimMarkowitz_InvestmentRules1989}
G.W. Kim and H.M. Markowitz, \textit{Investment rules, margin and 
market volatility}, J. Portfolio Manag. 16 (1989), 45–52.

\bibitem{LeBaron_AgentBased2006}
B. LeBaron, \textit{Agent-based computational finance},
Handbook of computational economics, Chapter 24, 1188-1233,
Elsevier, 2006.

\bibitem{LLS_MicroscopicMarket2000}
H. Levy, M.	Levy, S. Solomon,	
Microscopic Simulation of Financial Markets: 
From Investor Behavior to Market Phenomena, 2000.

\bibitem{LoMa_NonRandom2001}
A. Lo and C. MacKinlay, A Non-random Walk Down Wall Street.
Princeton Paperbacks, 2001.


\bibitem{LuiChong_TA2013}
K.M. Lui and T.T.L Chong, \textit{Do Technical Analysts Outperform Novice Traders: Experimental Evidence}, 
Economics Bulletin. 33(2013), No.4, 3080-3087.



\bibitem{Merton_Option1973}
R. Merton, \textit{Theory of Rational Option Pricing}, 
Bell Journal of Economics and Management Science. The RAND Corporation. 
4 (1973), No. 1, 141–183.

\bibitem{Meyers_Dynamics2011}
Meyers, R.A. (editor) Mathematics of Complexity and Dynamical
Systems. Springer, 2011 edition.


\bibitem{Oksendal-SDE2003}
B. \O ksendal, Stochastic Differential Equations. Sixth Edition. 
Universitext. Springer-Verlag, 2003.




\bibitem{SZSL_AgentBased2007}
E. Samanidou, E. Zschischang, D. Stauffer and T. Lux,
\textit{Agent-based models of financial markets},
Reports on Progress in Physics, 70 (2007) 409–450.

\bibitem{Samuelson_Random1965}
P. Samuelson, \textit{Proof That Properly Anticipated Prices Fluctuate Randomly},
Industrial Management Review, 6 (1965), 41–49.


\bibitem{SK_StochasticAveraging2015}
D.J. Simpson, R. Kuske,
\textit{Stochastic perturbations of periodic orbits with sliding}, 
J. Nonlinear Sci. 25 (2015), no. 4, 967–1014. 


\bibitem{TankovCont_Jump2003}
P. Tankov, and R. Cont, Financial Modeling with Jump Processes, 
Chapman and Hall/CRC, 2003.

\bibitem{TWG_Chartists2015}
F. Tramontana, F. Westerhoff, L. Gardini,
\textit{A simple financial market model with chartists and 
fundamentalists}, Mathematics and Computers in Simulation,
Volume 108  (2015), Issue C, Pages 16-40. 

\bibitem{Voit_StatisticalFinance2010}
J. Voit,  The Statistical Mechanics of Financial Markets,
3rd edition, 2010.

\bibitem{Westerhoff_Agentbased2009}
F. Westerhoff, \textit{A Simple Agent-based Financial Market Model: 
Direct Interactions and Comparisons of Trading Profits},
in: G.I. Bischi et al. (eds.), Nonlinear Dynamics in Economics,
Finance and the Social Sciences, Springer-Verlag, 2010,
pp 313-332.


\bibitem{ZungThien_Stochastic2015}
N.T. Zung, N.T. Thien, \textit{Reduction and integrability of
stochastic dynamical systems}, Fundam. Prikl. Mat. 
(in Russian) 20 (2015), no. 3, 213–249; English version:
arXiv:1410.5492.

\bibitem{Zung-AA2017}
N.T. Zung, A conceptual approach to the problem of 
action-angle variables, preprint arxiv:1706.08859 

\end{thebibliography}
\end{document}